\documentclass{acm_proc_article-sp}

\usepackage{algorithmic}
\usepackage{algorithm}

\usepackage{amssymb}
\usepackage{graphicx}
\usepackage{url}

\usepackage{amsmath}
\usepackage{amssymb}
\usepackage{amsfonts}
\usepackage{float}
\usepackage{graphics}
\usepackage{theorem}
\usepackage{euscript}
\usepackage{psfrag}
\usepackage{bbm}
\usepackage{todonotes}

\usepackage{flushend}

\newcommand{\set}[1]{\left\{#1\right\}}

\newcommand{\eps}{\varepsilon}

\newtheorem{thm}{Theorem}

\newtheorem{lem}[thm]{Lemma}
\newtheorem{prop}[thm]{Proposition}
\newtheorem{defn}{Definition}

\def\E{{\mathbb E}}        

\def\secSpace{0cm}
\def\parSpace{-0.08cm}

\newcommand{\comment}[1]{}

\usepackage{etoolbox} \makeatletter \patchcmd{\maketitle}{\@copyrightspace}{}{}{} \makeatother 

\begin{document}

\title{Budget-Constrained Item Cold-Start Handling in Collaborative Filtering Recommenders via Optimal Design }

%
%
%
%
%

\numberofauthors{6} 
%
\author{
%
%
\alignauthor
Oren Anava$^*$\\
       \affaddr{Technion, Haifa, Israel}\\
       \email{oanava@tx.technion.ac.il}
\alignauthor
Shahar Golan\\
       \affaddr{Yahoo Labs, Haifa, Israel}\\
       \email{shaharg@yahoo-inc.com}
\alignauthor Nadav Golbandi\\
       \affaddr{Yahoo Labs, Haifa, Israel}\\
       \email{nadavg@yahoo-inc.com}
\and  
\alignauthor Zohar Karnin\\
       \affaddr{Yahoo Labs, Haifa, Israel}\\
       \email{zkarnin@yahoo-inc.com}
\alignauthor Ronny Lempel\thanks{This work was done while the authors were at Yahoo Labs.} \\
       \affaddr{Outbrain Inc., Netanya, Israel}\\
       \email{rlempel@outbrain.com}
\alignauthor Oleg Rokhlenko\\
       \affaddr{Yahoo Labs, Haifa, Israel}\\
       \email{olegro@yahoo-inc.com}
\and
\alignauthor Oren Somekh\\
      \affaddr{Yahoo Labs, Haifa, Israel}\\
     \email{orens@yahoo-inc.com}
}

\maketitle
\begin{abstract}
It is well known that collaborative filtering (CF) based recommender systems provide better modeling of users and items associated with considerable rating history. The lack of historical ratings results in the user and the item cold-start problems. The latter is the main focus of this work. Most of the current literature addresses this problem by integrating content-based recommendation techniques to model the new item. However, in many cases such content is not available, and the question arises is whether this problem can be mitigated using  CF techniques only. We formalize this problem as an optimization problem: given a new item, a pool of available users, and a budget constraint, select which users to assign with the task of rating the new item in order to minimize the prediction error of our model. We show that the objective function is monotone-supermodular, and propose efficient optimal design based algorithms that attain an approximation to its optimum. Our findings are verified by an empirical study using the Netflix dataset, where the proposed algorithms outperform several baselines for the problem at hand. 
\end{abstract}

\vspace{-0.2cm}
\vspace{-0.2cm}
\vspace{-0.2cm}

\section{Introduction}
\label{intro}
Recommendation technologies are increasingly being used to route relevant and enjoyable information to users. Whether they try to navigate through overwhelming Web content, choose a restaurant, or simply find a book to read, users  find themselves being guided by modern online services that use recommendation systems.  Usually, such services are based on users' feedback and stated preferences, and use editors, content analysis or wisdom of the crowds to tag their database items.

One of the most common and effective recommendation techniques is Collaborative filtering (CF). This technique relies only on past user behavior (e.g., previous transactions or feedback), and does not require the creation of explicit profiles. Moreover, it requires no domain knowledge or content analysis, and excels at exploiting popularity trends, which drive much of the observed user interaction. 
One of the fundamental problems arising when employing CF techniques is the \textit{cold-start} problem. Roughly speaking, the cold-start problem is caused by the system's incapability of dealing with new items or new users due to the lack of relevant transaction history. 

This work focuses on the item cold-start problem, \textit{without} assuming any availability of context nor content-based information. In particular, we consider the following setting: a publisher has a set of users, and is aware of those users' historical ratings of items (movies, books, etc.). Now the publisher is interested in evaluating a new item, on which she has no prior knowledge. At her disposal is a pool of available users, from which she can assign $B$ users to the task of rating the new item. After receiving their ratings, the publisher estimates, to the best extent possible, how the remaining users will rate the new item. 

In order to mitigate the reviewing period, the $B$ users are all selected at once - the publisher does not have the luxury to receive some ratings, and then to adaptively select additional users based on that signal. An important aspect of this problem setting is that the satisfaction of the assigned users does not factor into the equation - there is no cost associated with selecting a user who will dislike the item. The sole purpose is to select  users who will provide best characterization of the new item.

The rest of this paper is organized as follows: the remaining of this section presents informal statement of our results, and surveys the related work. Section~\ref{pre} introduces notations and concretely defines the problem. Section~\ref{optimal} then presents our optimal design based algorithms and analysis.
Experimental results are reported in Section~\ref{sec:results}. We conclude and discuss future work in Section~\ref{sec:conc}.

\subsection{Informal Statement of Results} \label{informal}
We adopt a common assumption in \textit{Matrix Factorization} (MF)-based CF \cite{Koren-Bell-Volinsky_IEEE-Comp-2009}:  ratings of items by users can be well approximated by the product of two low-dimensional matrices, one representing low-dimensional latent vectors of users, and the other holding latent representations of items. The problem at hand then translates to the following mathematical abstraction: given latent vectors representing each of the users, and a new item whose latent vector is unknown, choose which $B$ ratings to reveal so that the item's latent vector could be estimated most accurately.

The main contributions of this work are as follows: (1) We cast this budget-constrained user selection problem as an optimization problem, in which the objective function stands for the expected prediction error of our model. (2) We devise two novel optimal design based algorithms for the problem, each attains approximation to the optimum under certain assumptions. (3) We validate the theoretical results with an empirical study, by simulating the setting of the problem using the Netflix dataset and comparing the effectiveness of our algorithms to that of several baselines.
%

\vspace{\parSpace}

\subsection{Related Work}
\label{sec:related}


 The two major fields of CF are \textit{neighborhood methods} and \textit{Latent Factor Models} (LFM). 
%
%
 LFM characterizes both items and users as vectors in a space automatically inferred from observed data patterns. The latent space representation strives to capture the semantics of items and users, which drive most of their observed interactions. One of the most successful realizations of LFM, which combines good scalability with predictive accuracy, is based on low-rank MF  (e.g., see \cite{Koren-Bell-Volinsky_IEEE-Comp-2009}). In particular, low-rank MF provides a substantial expressive power that allows modeling specific data characteristics such as temporal effects \cite{koren_CommACM2010}, item taxonomy \cite{Dror-Koenigstein-Koren_RecSys2011}, and attributes \cite{Agarwal-Chen_KDD2009}.

One of the inherent limitations of MF is the need for historical user-item interactions. When such history is limited, MF-based recommenders cannot reliably model new entities, leading to the item and user cold-start problems. Despite the fact that users and items are usually represented similarly in the latent space, these two problems are essentially different due to the ability  to interview new users when joining the service in order to bootstrap their modeling. Moreover, each of these problems has a different line of previous works, as surveyed below.\\


\textit{\textbf{User cold-start problem.}}
Modeling the preferences of new users can be done most effectively by asking them to rate several carefully selected items of a seed set during a short interview \cite{Kohrs-Merialdo_ICMF2001,Rashid-Albert-Cosley-Lam-McNee-Konstan-Riedl_IUI2002,Rashid-Karypis-Riedl_KDD2008,Golbandi-Koren-Lempel_CIKM2010}. Item seed sets were constructed according to various criteria such as \textit{popularity} (items should be known to the users), \textit{contention} (items should be indicative of users' tendencies), and \textit{coverage} (items should possess predictive power on other items).

The importance of adaptive interviewing process using adaptive seed sets was already recognized in \cite{Rashid-Albert-Cosley-Lam-McNee-Konstan-Riedl_IUI2002}. Adaptive interviewing is commonly implemented by decision trees algorithms \cite{Rashid-Karypis-Riedl_KDD2008,Lee_ESA2010,Golbandi-Koren-Lempel_WSDM2011}. 
A method for solving the problem of initial interview construction within the context of learning user and item profiles is presented in \cite{Zhou-Yang-Zha_SIGIR2011}.\\

\textit{\textbf{Item cold-start problem.}}
To mitigate the item-cold problem of CF, a common practice involves utilizing external content on top of users' feedback. The basic approach is to leverage items' attributes and combine them with the CF model in a way that allows treating new items
\cite{Agarwal-Chen_KDD2009,Gunawardana-Meek_RECSYS2008,Gunawardana-Meek_RECSYS2009,Park-Chu_RECSYS2009}. In \cite{Agarwal-Chen_KDD2009} the authors proposed a regression-based latent factor model for cold-start recommendation. 
In their model a dyadic response matrix $Y$ is estimated by a latent factor model $Y \approx U^\top V$, where the latent factor matrices, $U$ and $V$, are estimated by regression $U \approx FX$ and $V\approx MZ$. The matrices $X$ and $Z$ are the user attribute and item feature matrices, and $F$ and $M$ are weight matrices learnt by regression. In a later work \cite{Park-Chu_RECSYS2009}, the authors improve the performance achieved in \cite{Agarwal-Chen_KDD2009} by solving a convex optimization problem that estimates the weight matrices, instead of computing the low-rank matrix decomposition of \cite{Agarwal-Chen_KDD2009}.
Another approach to tackle the item cold-start issue by combining content information with CF,  based on \textit{Boltzmann machines},  is studied in \cite{Gunawardana-Meek_RECSYS2008,Gunawardana-Meek_RECSYS2009}. 

The  works described above require content or context data for new items. However, such data may not be available. In such cases, ratings for new items are arriving and one wishes to predict the missing ratings, with increasing accuracy, as more feedback is provided. Due to the incremental nature of the problem, one cannot afford to retrain the CF model with each arriving rating. A common heuristic approach is to estimate a new item's latent vector by a linear combination of the latent vectors of its raters and their respective ratings. It is based on the fact that users and items coexist as vectors in the same latent space and that ratings are usually estimated by user-item vector similarities. This principle was successfully practiced in various MF techniques \cite{Paterek_KDD2007,Koren_KDD2008,Aharon-Kagian-Koren-Lempel_RecSys2012,Aizenberg-Koren-Somekh_WWW2012}.
We note that this approach will serve us as baseline, as it does not require any content-based information and is thus suitable for our setting. \\

\textit{\textbf{The optimal design approach.}}
In the  statistical field called {\it optimal design of experiments}, or just {\it optimal design} \cite{atkinson1992optimum,Wu78}, a statistician is faced with the task of choosing a limited amount of experiments to perform from a given pool, with the goal of gaining the most informative results. There are various measurements for the gain obtained from a subset of experiments. The measure of interest in our setting is referred to as the $A$-optimality criterion. According to this criterion, the experiments correspond to noisy outputs of a (usually linear) regression function and the goal is to minimize the Mean Squared Error (MSE) over all the possible inputs of the regressor. 


An early attempt to tackle the user cold-start problem using optimal design techniques and an interviewing process appears in \cite{Yu-Bi-TrespI_CML2006}.
In this work, an adaptation of $A$-optimality, also referred to as \emph{transductive optimal design}, aims to minimize the MSE over some arbitrary (yet known in advance) set of points (or users, in our context).
Despite the original presentation in the setting of the user cold-start problem, the generality of the employed techniques allows handling the item cold-start problem in a similar manner.
 However, their work differs from ours due to the absence of performance guarantee for the proposed algorithms, and the use of an inferior prediction model.

\section{Preliminaries and Model}
\label{pre}

%



\subsection{The Latent Factor Model}
We denote by  $\mathcal{I}$ the set of all items, and by $\mathcal{U}$ the set of all users. The cardinality of these sets is denoted by  $ \left| \mathcal{I} \right| = m$, and $ \left| \mathcal{U} \right| = n$. We use $\mathcal{R}$ to denote the rating matrix, which is of size $n \times m$. A rating $r_{ui}$ indicates the preference of item $i$ by user $u$, where high values mean stronger preference. The matrix $\mathcal{R}$  is typically sparse, since most users rate just a small portion of the items.
%
%
%

Using the notation above, the LFM (with $k$ denoting its dimension) seeks to estimate each rating ${r}_{ui}$ as follows:
\begin{equation}
{r}_{ui} \approx \mu + b_i + b_u +  Q_i ^\top P_u \  ,
\label{eq:LFM}
\end{equation}
where $b_i \in  \mathbb{R}$ and $Q_i \in  \mathbb{R}^k$ are the bias and the latent factor vector of item $i$, respectively; $b_u  \in  \mathbb{R}$ and $P_u  \in  \mathbb{R}^k$ are the bias and the latent factor vector of user $u$, respectively; and $\mu \in \mathbb{R}$ denotes the overall average rating in $\mathcal{R}$. We use $\varepsilon_{ui}$ to denote the residual error of our model with respect to user $u$ and item $i$.

Intuitively, the term $Q_i ^\top P_u$ captures the interaction (or affinity) between user $u$ and item $i$, where high values implies stronger preference and vice versa. We denote by $Q$ the $k \times m$ matrix, whose columns correspond to the vectors $Q_i$ for $i=1,\ldots,m$, and by $P$ the $k \times n$ matrix, whose columns correspond to the vectors $P_u$ for $u=1,\ldots,n$. 

An important aspect of our work is the following: as training the LFM is a completely orthogonal task to ours, we assume that we are given an already trained model and refer to it as ground truth. Practically, we assume that  each rating ${r}_{ui}$ complies with the following noisy model:
\begin{equation}
{r}_{ui} = \mu + b_i + b_u +  Q_i ^\top P_u + \varepsilon_{ui} \ ,
\label{eq:LFM_ground}
\end{equation}
where $\varepsilon_{ui}$ is assumed to be a zero-mean noise term. This assumption is very common and can be found in several previous works (see \cite{Yu-Bi-TrespI_CML2006} for instance). Moreover, a simple sanity check can verify that the assumption complies with the employed LFM: test the significance of the residuals as generated from a Gaussian distribution. 


\subsection{Problem Definition}
We proceed to formally define our problem. Let $i$ be a new item ($i \not\in \mathcal{I}$), and denote by $\mathcal{U}^i \subset \mathcal{U}$ the pool of available users to rate it.
 Also, we denote the budget constraint by $B$ (i.e., $B$  is the number of users we are allowed to assign with the task of rating item $i$), and use the notation $\mathcal{U}^i_B $ for subsets of $\mathcal{U}^i$, which are of size $B$.



%

Now, given a budget constraint $B$, our goal is to select $B$ users to rate  item $i$ in order to minimize the expected MSE on the set of the remaining users $\mathcal{U}  \setminus \mathcal{U}^i_{B}$. More formally, we wish to solve the following optimization problem:
\begin{equation*} 
\min_{ \mathcal{U}^i_{B} \subset  \mathcal{U}^i} \left\{ \mathbb{E} \bigg[  \frac{1}{\left| \mathcal{U} \setminus \mathcal{U}^i_{B} \right|}   \sum_{u \in  \mathcal{U}  \setminus \mathcal{U}^i_{B} } \left( \tilde{r}_{ui} - {r}_{ui} \right)^2 \bigg] \right\} ,
 \end{equation*}
 where $ \tilde{r}_{ui}$ denotes our prediction of $ {r}_{ui}$, and the expectation is taken over the noise terms consisting the ratings $\set{{r}_{ui}}$ (as defined in Equation~\eqref{eq:LFM_ground}).

 To simplify the notations, in the sequel we refer to the set of the remaining users as $\mathcal{U}$ (instead of $\mathcal{U} \setminus \mathcal{U}_B^i$). 
 Essentially,  this means that the sets $\mathcal{U}$ and $\mathcal{U}_B^i$ are disjoint, and the problem of interest is thus translated to:
 \begin{equation} \label{prob_def}
\min_{ \mathcal{U}^i_{B} \subset  \mathcal{U}^i} \left\{ \mathbb{E} \bigg[  \frac{1}{\left| \mathcal{U} \right|}   \sum_{u \in  \mathcal{U} } \left( \tilde{r}_{ui} - {r}_{ui} \right)^2 \bigg] \right\} .
 \end{equation}
Clearly, the objective function inherently depends on how we generate the predictions $\set{\tilde{r}_{ui}}$ given the set of  $B$ users 
$\mathcal{U}^i_{B}$, and their ratings of item $i$.
Thus we can divide the problem in Equation \eqref{prob_def} into two distinct problems:
 
\begin{description}
\item[Rating Prediction (Problem A):] given a subset of users $\mathcal{U}^i_B \subset \mathcal{U}^i$ and their ratings of item $i$, generate predictions $\set{\tilde{r}_{ui}}$ for all $ u \in \mathcal{U}$.
\item[User Selection (Problem B):] given a pool of available users $\mathcal{U}^i$, select a subset of users $\mathcal{U}^i_B \subset \mathcal{U}^i$ to reveal their ranking of item $i$.
\end{description}

We shall point out that viewing the latent factor model as ground truth gives rise to an optimal solution for \textbf{Problem A}: least squares estimation (detailed in Section~\ref{optimal}). However, in the CF literature there exist other approaches that can potentially function better in practice (especially if the data is not well represented by the LFM). These approaches are presented in Section~\ref{baseline}, and will serve us as baselines for evaluating the merit of our approach.

%
%

\section{Algorithms and Analysis}
\label{optimal}

In this section we tie together the two problems described in Section \ref{pre}, and present a holistic approach that tackles the unified problem as a whole - the optimal design approach. \textit{Optimal design} is a statistic paradigm  \cite{atkinson1992optimum,Wu78} that is aimed at selecting which experiments to conduct out of a given pool in order to maximize the result accuracy. In our setting, we wish to select a subset of users such that the prediction error of our model is minimized. Roughly speaking, the optimal design approach seeks to estimate the MSE (over the entire user set) for each selection of $B$ users, \emph{without} exposing their actual ratings. This way, we can evaluate the merit of any given subset of users (in terms of MSE), and make the selection accordingly.

Basically, the optimal design approach relies on the assumption that the ratings are generated via Equation~\eqref{eq:LFM_ground}, and consequently on the optimal solution for \textbf{Problem A}: least squares estimation. To make this point clearer we first introduce this estimator.

\subsection{Least Squares Estimation}
Assume we are given a new item $i \not\in \mathcal{I}$, and a subset of 
users   $\mathcal{U}^i_B \subset \mathcal{U}^i$ who provided their ratings for this item. Our task is to  predict $\tilde{r}_{ui}$ for all $u \in \mathcal{U}$. 
Recall our modeling assumption regarding each of ratings $r_{ui}$ from Equation \eqref{eq:LFM_ground}.
In particular, we assume that the ratings of the new item $i$ by the users within $\mathcal{U}^i_B $ follow this model as well, where  $ \mu $, $ b_u  $ and $  P_u  $ are given by our prediction model (Section \ref{pre}), whereas $ b_i  $ and $ Q_i $ are unknown to us (since the new item was not considered while training the model). 
Therefore, the least squares estimator seeks to estimate $ {b}_i $ and $ {Q}_i $, and consequently to generate a prediction
$$\tilde{r}_{ui} =  \mu + \tilde{b}_i  + b_u+  \tilde{Q}_i  P_u .$$
From now on, we denote by $ ( {b}_i , {Q}_i ) $ the $(k+1)$-dimensional concatenation of $b_i$ and the vector $Q_i$; the notation $ ( \tilde{b}_i , \tilde{Q}_i ) $ shall stand for the corresponding estimator.  

Using the notation above, the least squares estimator aims to estimate $( b_i , Q_i  ) $ by minimizing the MSE over the set $\mathcal{U}^i_B$, that is, to solve the following problem:
\begin{equation*} 
\min_{ \substack{b \in \mathbb{R} \\\ q \in  \mathbb{R}^k }}  \left\{ \frac{1}{ \left|  \mathcal{U}^i_B  \right|  }   \sum_{ v \in \mathcal{U}^i_B } \left(  r_{vi} - q^\top P_v- b - b_v - \mu \right)^2     
\right\}  .
\end{equation*}
This minimization problem is in fact analytically solvable, and yields the following solution:
\begin{equation*} 
\big(  \tilde{b}_i , \tilde{Q}_i \big) 
= \left(  \sum_{ v \in \mathcal{U}^i_B } {P}^\prime_v {P}_v^{\prime \top} \right)^{-1} \left( \sum_{ v \in \mathcal{U}^i_B } (r_{vi} - b_v - \mu )   {P}_v^{\prime} \right)   ,
\end{equation*}
where $   {P}^\prime_v $ is the concatenated column vector $  ( 1 ,  {P}_v ) $.

The above might  not be well-defined since $\sum_{ v \in \mathcal{U}^i_B } {P}^\prime_v {P}_v^{\prime \top}$ need not be invertible in general. In practice,  a regularization term of the form $\lambda \left( \| q \|_2^2 + b^2 \right)$ is usually added to the objective function to avoid this problem. The estimator $ ( \tilde{b}_i , \tilde{Q}_i  )$  then takes the following form:
 \begin{equation*} 
 \left( \lambda   I + \sum_{v \in \mathcal{U}^i_B } {P}^\prime_v {P}_v^{\prime \top} \right)^{-1}   \left( \sum_{v \in \mathcal{U}^i_B} (r_{vi} - b_v - \mu )   {P}_v^{\prime} \right) .
\end{equation*}
In the sequel, we will assume that $\sum_{ v \in \mathcal{U}^i_B } {P}^\prime_v {P}_v^{\prime \top}$ is invertible to ease the readability of the paper. We emphasize that all the  results presented in this paper would still hold for the regularized estimator.

\subsection{The Optimal Design Approach}
We proceed to formally define necessary notations that will help us adapting the optimal design approach to our setting.
We somewhat abuse notations and denote by $ P $ the matrix whose columns correspond to the latent factor vectors $ P_u^\prime $ for $ u \in \mathcal{U}  $, and by ${P}_B$ the matrix whose columns correspond to the latent factor vectors $ P_v ^\prime$ for $ v \in \mathcal{U}^i_B  $.
  Similarly, we use the notations $\varepsilon_{\mathcal{U}} $ and $\varepsilon_B$ for vectors whose elements correspond to $ \varepsilon_{ui} $ for $ u \in \mathcal{U}$ and $ \varepsilon_{vi} $ for $ v \in \mathcal{U}_{B}^i $, respectively. Finally, we use $r_B$ to denote the vector whose elements correspond to $(r_{vi} - b_v - \mu)$ for $ v \in \mathcal{U}_{B}^i $.

%
 We divide the analysis into two cases: (1) the noise terms $\set{\varepsilon_{ui}}$ are assumed to be zero-mean and i.i.d.; and (2) the noise terms $\set{\varepsilon_{ui}}$ are only assumed to be zero-mean and independent (but not necessarily identically distributed).

\subsubsection{Identically Distributed Noise Terms} \label{iid}
The following is our key observation: the optimization problem considered in Equation \eqref{prob_def} can be reduced to a simpler problem, if the noise terms are assumed to be i.i.d. and the least squares estimator is used for estimating  $ ( b_i , Q_i ) $. This observation is stated and proven below in Lemma \ref{TRACE}.

Before we continue, we will assume without loss of generality that $ P P ^\top = \left| \mathcal{U} \right|   I_{(k+1) \times (k+1)}$, meaning that the user vectors are in isotropic position. Otherwise , we can simply apply an invertible linear transformation to the LFM space of the users and its inverse to the space of the items; this does not affect our results as $(F ^\top x)^\top (F^{-1}y) =  x^\top y $ for all vectors $x,y$ and invertible linear transformation $F$. 
This assumption is here merely to simplify the statements and proofs; the statements hold without it as well, and in particular, one is not required to compute the inverse of $P P ^\top$ when implementing the algorithm. 

Additionally, we assume that for any $ \mathcal{U}^i_{B}  \subset  \mathcal{U}^i$, it holds that the matrix $P_B P_B^{\top}$ is diagonal. While seeming like a non-realistic assumption, it captures the nature of the problem and facilitates the analysis. Concretely, this assumption will later be used to prove that our objective is supermodular (to be later defined). The problem can be solved without this assumption by using weak supermodularity property, as proposed in \cite{boutsidis2015greedy}.

%

 \begin{lem} \label{TRACE}
Assume there exist $ b_i $ and  $  Q_i  $, such that for every $u \in \mathcal{U}$ it holds that  $r_{ui} =  \mu + b_u + b_i  +  Q_i^{ \top} P_u  + \varepsilon_{ui}$, where $ \mathbb{E} \left[ \varepsilon_{ui} \right] =0 $ and $ \mathbb{E} \left[ \varepsilon_{ui}^2 \right] = \sigma^2 $. Then, the following is equivalent\footnote{The equivalence here between the two problems is in the following sense: the value of both objective functions is the same for any subset $ \mathcal{U}^i_{B}  \subset  \mathcal{U}^i$.}   to the problem presented in Equation \eqref{prob_def}:
$$
 \min_{ \mathcal{U}^i_{B} \subset  \mathcal{U}^i}  \left\{    \sigma^2    \text{tr} \Big( \left( P_B P_B^{\top} \right)^{-1}  \Big) +\sigma^2   \right\},
 $$
 if the least squares estimator is considered.
 \end{lem}

The expression in the lemma above  has the following interpretation: the additive $\sigma^2$ term is the inherent model error. I.e., this is the term that cannot be avoided as long as we assume a latent factor model of dimension $k$. The second term of $\sigma^2    \text{tr} \big( \big( P_B P_B^{\top} \big)^{-1}  \big)$ represents the error originating from a sub-optimal choice of the item's parameters, i.e.,\ from the distance between $( \tilde{b}_i , \tilde{Q}_i  ) $ and $( b_i , Q_i ) $.  
The proof of the lemma appears in Appendix \ref{app: proof of lemma 1}.\\

After establishing a new (and equivalent) optimization problem, we turn now to present a simple greedy algorithm for finding an approximation to its optimum. Here, $P_{B \setminus v_j} $ refers to the matrix $P_B$, where the column that  corresponds to the user ${v_j}$ is removed.
\begin{algorithm}[tb]
\caption{Backward Greedy Selection (BGS1)}
\label{alg2}
\begin{algorithmic}[1]
\STATE Input: users set $\mathcal{U}^i$, and corresponding matrix $P_{\mathcal{U}^i}$.
\STATE Output: users subset ${\mathcal{U}_B^{ALG}}$.
\STATE Initialize $P_{B}=P_{\mathcal{U}^i}$ and $\mathcal{U}_B^{ALG} =\mathcal{U}^i$.
\FOR {$j=1$ to $ | \mathcal{U}^i | - B$}
\STATE 
$v_j \leftarrow \arg\min_{v \in \mathcal{U}_B^{ALG} } \left\{  \text{tr} \left( \left( P_{B \setminus v}  P_{B \setminus v} ^{ \top} \right)^{-1}   \right)  \right\} 
$
\STATE Update $P_{B} \leftarrow P_{B \setminus v_j} $.
\STATE Update $\mathcal{U}_B^{ALG} \leftarrow \mathcal{U}_B^{ALG} \setminus v_j $.
\ENDFOR
\end{algorithmic}
\end{algorithm}
Before stating our main theorem, we introduce some necessary definitions:
\begin{defn}
let $f:\mathbb{R}_{+} \to \mathbb{R}$. For a positive definite matrix $A \in R^{k \times k}$ with the decomposition
$ A = U^\top \text{diag}(\lambda_1,\ldots,\lambda_k) U $,
we extend the definition of  $f$ as follows:
$$ f(A) = U^\top \text{diag}(f(\lambda_1),\ldots,f(\lambda_k)) U .$$
\end{defn}
\begin{defn}
We say that $f$ is \emph{operator monotone} if for any positive definite matrices $A$ and $B$:
$$  
A \preceq B \Rightarrow f(A) \preceq f (B) .
$$
\end{defn}
\begin{defn}
We say that $f$ is \emph{operator antitone} if $-f$ is operator monotone.
\end{defn}
\begin{defn}
We say that $f:2^{E} \rightarrow \mathbb{R} $ is \emph{supermodular} if 
$$ f(A \cup \left\{x\right\} ) - f(A) \leq  f(B \cup \left\{x\right\} ) - f(B) ,$$
for any $A \subset B \subset E$ and $x \in E \setminus B$.
\end{defn}
\begin{defn}
The steepness of a decreasing supermodular function $f$ is denoted by $s$ and defined as follows: 
$$
s= \max_{x \in E} \frac{  \left[  f(\phi)-f(x) \right] -  \left[ f(E \setminus \left\{x\right\} ) - f(E) \right]  }{  \left[  f(\phi)-f(x) \right]  } .
$$
\end{defn}
In the last definition, if $ f(\phi)$ is not defined we can simply extend it as follows:
\begin{equation} \label{extension}
 f(\phi) = \max_{ \substack{A,B \subset E \\\ A \cap B = \phi }}  \left\{ f(A) + f(B)  -  f(A \cup B)  \right\} 
\end{equation}

%
%


\noindent The following is our main theorem:
\begin{thm} \label{GREEDY}
Let $i$ be a new item, and $\mathcal{U}^i $ its corresponding set of available users, and assume that  $ \mathbb{E} \left[ \varepsilon_{ui} \right] =0 $ and $ \mathbb{E} \left[ \varepsilon_{ui}^2 \right] = \sigma^2 $ for all $u \in \mathcal{U}$. Then, Algorithm \ref{alg2} generates a subset $\mathcal{U}^{ALG}_B \subset  \mathcal{U}^i$ for which:
\begin{equation*} \label{approx}
\E \left[ \text{MSE} \big(  \mathcal{U}^{ALG}_B \big) \right] \leq   \sigma^2 + \frac{ e^t - 1  }{t}   \left( \E \left[  \text{MSE} \big(  \mathcal{U}^{*}_B \big) \right] - \sigma^2 \right)   .
\end{equation*}
where $t=\frac{s}{1-s}$ and ${\cal U}_B^*$ is the optimal subset of $B$ users, minimizing the expected MSE.
\end{thm}

The above statement can be interpreted as follows. Both Algorithm \ref{alg2} and the optimal solution have an additive term of $\sigma^2$, corresponding to the inherent LFM error. The additional term, corresponding to the sub-optimality of the choice of $( \tilde{b}_i , \tilde{Q}_i  ) $, is multiplicatively approximated by a factor of $(e^t-1)/t$ compared to the optimal selection. Notice that the stated guarantee strictly dominates the following one: 
$$\E \left[ \text{MSE} \big(  \mathcal{U}^{ALG}_B \big) \right] \leq  \frac{ e^t - 1  }{t}   \E \left[  \text{MSE} \big(  \mathcal{U}^{*}_B \big) \right] ,$$
since $\sigma^2 >0$ (trivially holds). We shall proceed to prove the main theorem.

\begin{proof}
We first define an auxiliary function $\Phi : 2 ^{\mathcal{U}^i } \rightarrow \mathbb{R}$ as follows:
$$ \Phi ( \mathcal{U}^i_{D} ) = \text{tr} \left(\left( P_D P_D^{ \top} \right)^{-1}  \right) - \varphi_i ,$$ where $\varphi_i =  \text{tr} \big(\left( P_{\mathcal{U}^i} P_{\mathcal{U}^i}^{ \top} \right)^{-1}  \big)$, and  $ \Phi ( \varnothing )$ is defined as in Equation~\eqref{extension}. 
Notice that $\Phi$ is well-defined for any subset $\mathcal{U}_D^i \subset \mathcal{U}^i$ (of size $0\leq D \leq  \left| \mathcal{U}^i \right| $), and not only for subsets of size $B$. The matrix ${P}_D$ is then defined as the matrix whose columns correspond to the latent factor vectors $ P_v ^\prime$ for $ v \in \mathcal{U}^i_D  $.

Next, we show that $ \Phi $ has the following three properties: it is supermodular, monotonically  decreasing, and equal to 0 at $ \mathcal{U}^i $. 
For such functions, \cite{Ilev01} proved that the backward greedy algorithm, eliminating elements one by one, attains an $\big( \frac{ e^t - 1  }{t} \big)$-approximation to the optimum. 
The third property follows immediately from the definitions of $\Phi$ and $\varphi_i$, and thus we are left to prove the other two properties.

Before proving these properties, we provide a simple intuition. 
Notice first that $ \E \left[ \text{MSE} (\mathcal{U}_D^i ) \right] $ can be written as a linear function of $\Phi (\mathcal{U}_D^i )$:
$$ \E \left[ \text{MSE} (\mathcal{U}_D^i ) \right] = \sigma^2    \left( \Phi (\mathcal{U}_D^i ) + \varphi_i + 1 \right).$$
Thus, minimizing $\Phi$ is equivalent to minimizing the MSE. Now, the monotonicity of $\Phi$ gives rise to the following insight: the MSE decreases as the number of users increases. The supermodularity then translates into the following rational: the marginal return of a user diminishes as the set of selected users expands.

We shall proceed to formally prove these properties.
For the monotonicity, we take two subsets $\mathcal{U}_D^i , \mathcal{U}_E^i \subset \mathcal{U}^i$, such that $\mathcal{U}_D^i \subset \mathcal{U}_E^i $, and prove that $ \Phi ( \mathcal{U}_D^i ) \geq \Phi ( \mathcal{U}_E^i )$.  Thus, denote by $P_D$ and $P_E$ the matrices that correspond to $\mathcal{U}_D^i$ and $\mathcal{U}_E^i$, respectively. Then, it holds that $ P_E P_E^\top \succeq P_D P_D^\top$, which also implies that 
$$ \left( P_E P_E^\top \right)^ {-1} \preceq \left( P_D P_D^\top \right)^ {-1} .$$
Now, since the trace of a positive semidefinite (PSD) matrix is the sum of its eigenvalues, the above also implies that 
$$ \text{tr} \left( \left( P_E P_E^\top \right)^ {-1} \right) \leq \text{tr} \left( \left( P_D P_D^\top \right)^ {-1} \right) ,$$
which proves that $\Phi$ is monotonically decreasing.


For the supermodularity, we use recent techniques presented in \cite{sagnol2013approximation}, and specifically the following proposition:
\begin{prop} 
Let  $\mathcal{U}_D^i  \subset  \mathcal{U}^i $, and $P_D$ be the matrix whose columns correspond to $P_v^{\prime}$ for $v \in \mathcal{U}_D^i $. Then, for $f:\mathbb{R}_{+} \to \mathbb{R}$ such that $f^\prime$ is operator monotone, it holds that 
$$ \Phi ( \mathcal{U}_D^i  ) = \text{tr} \left( f \big( P_D P_D^\top \big) \right) $$
is supermodular.
\end{prop}

In our case $f(x) = \frac{1}{x}$, and its derivative  $f^\prime (x) = - \frac{1}{x^2}$ is operator monotone for diagonal matrices, which implies that $ \Phi$ is supermodular.
The proof resembles the proof of Corollary 2.4 of  \cite{sagnol2013approximation}, albeit plugging operator monotone derivative instead of operator antitone one.

Finally, as mentioned before, these three properties yield an $\big( \frac{ e^t - 1  }{t} \big)$-approximation to the problem of minimizing $\Phi$. That is, Algorithm \ref{alg2} generates a subset $ \mathcal{U}_{B} ^{ALG}$ for which:
\begin{align*}
 \Phi \big( \mathcal{U}_{B} ^{ALG} \big) & \leq \left( \frac{ e^t - 1  }{t} \right)  \min_{ \mathcal{U}^i_{B} \subset  \mathcal{U}^i}   \Phi \big( \mathcal{U}_{B} ^{i} \big)  \\
 & =   \left( \frac{ e^t - 1  }{t} \right)  \min_{ \mathcal{U}^i_{B} \subset  \mathcal{U}^i} \left\{ \frac{ \E \left[ \text{MSE} (\mathcal{U}_B^{i}) \right] }{ \sigma^2 } - 1 - \varphi_i  \right\}  \\
 & \leq  \left( \frac{ e^t - 1  }{t} \right)   \min_{ \mathcal{U}^i_{B} \subset  \mathcal{U}^i} \left\{ \frac{ \E \left[ \text{MSE} (\mathcal{U}_B^{i}) \right] }{ \sigma^2 } -1 \right\}  - \varphi_i   \\
 & =   \left( \frac{ e^t - 1  }{t} \right)    \left( \frac{ \E \left[ \text{MSE} (\mathcal{U}_B^{*}) \right] }{ \sigma^2 } - 1 \right) - \varphi_i   .
\end{align*}
By substituting $\Phi (\mathcal{U}_B^{ALG} )  = \frac{ \E \left[ \text{MSE} (\mathcal{U}_B^{ALG}) \right] }{ \sigma^2 } - \varphi_i - 1 $ in the inequality above,
 we get the stated result.
%
\end{proof}

%
%
%
\subsubsection{Non-Identically Distributed Noise Terms} \label{niid}
The analysis in Section \ref{iid}  relies heavily on the assumption that $\set{\varepsilon_{ui}}$ are identically distributed. However,  it is sometimes reasonable to assume that each user has a different noise distribution with respect to the employed model. In such cases, the identical distribution assumption only weakens our model. In this section we extend our analysis to the case where each of the users has a different  (yet  known in advance) noise distribution. 

Thus, let $i$ be a new item, and $ \mathcal{U}_{B}^i \subset  \mathcal{U}^i $ a subset of users assigned for the task of rating it.
We start by introducing a new estimator for $ (b_i , Q_i ) $, which generalizes the least squares estimator:
$$ \big( \tilde{b}_i , \tilde{Q}_i  \big)  
 = \left( P_B C_B^{-2} P_B^{ \top} \right)^{-1}  P_B C_B^{-2}  r_B  , $$
where again, $ C_B $ is the square root of the covariance matrix that corresponds to $ \varepsilon_{vi} $ for $ v \in \mathcal{U}^i_{B} $. Intuitively, this estimator is a least squares estimator that considers variance-scaled users. In the sequel, we refer to this estimator as \emph{generalized least squares estimator}. The following lemma states that this is indeed an unbiased estimator for $ (b_i , Q_i ) $. The proof is technical and appears in Appendix \ref{app: proof of lemma 4}.
\begin{lem} \label{SECOND}
Assume there exist $ b_i $ and  $  Q_i  $, such that for every $u \in \mathcal{U} $ it holds that  $r_{ui} =  \mu + b_u + b_i  +  Q_i^{ \top} P_u  + \varepsilon_{ui}$, where $ \mathbb{E} \left[ \varepsilon_{ui} \right] =0 $ and $ \mathbb{E} \left[ \varepsilon_{ui}^2 \right] = \sigma_u^2 $. Then,
\begin{equation} \label{generalized}
 \big( \tilde{b}_i , \tilde{Q}_i  \big)  
 = \left( P_B C_B^{-2} P_B^{ \top} \right)^{-1}  P_B C_B^{-2}  r_B  
 \end{equation}
is an unbiased estimator for $ (b_i , Q_i ^\top)$. 
\end{lem}
We proceed to state the lemma that encapsulates our parallel key observation for the case of non-identically distributed noise terms; its proof appears in Appendix \ref{app: proof of lemma 5}.
 
 \begin{lem} \label{trace2}
Assume there exist $ b_i $ and  $  Q_i  $, such that for every $u \in \mathcal{U}$ it holds that  $r_{ui} =  \mu + b_u + b_i  +  Q_i^{ \top} P_u  + \varepsilon_{ui}$, where $ \mathbb{E} \left[ \varepsilon_{ui} \right] =0 $ and $ \mathbb{E} \left[ \varepsilon_{ui}^2 \right] = \sigma_u^2 $. Then, the following is equivalent  to the problem presented in Equation \eqref{prob_def}:
$$
 \min_{ \mathcal{U}^i_{B} \subset  \mathcal{U}^i}  \left\{  \text{tr} \left( \left( P_B C_B^{-2} P_B^{\top} \right)^{-1}  \right) +  \frac{1 }{ \left| \mathcal{U} \right| }  \sum_{u \in \mathcal{U} } \sigma_u^2 \right\},
 $$
 if the generalized least squares estimator is considered.
 \end{lem}

After establishing here also an equivalent optimization problem, we turn to present a simple adaptation of Algorithm \ref{alg2} to this problem. Here, the notation
$ C_{\mathcal{U}^i} $ refers to the square root of the covariance matrix of all available users, and the notations $P_{B \setminus v_j} $ and $C_{B \setminus v_j} $ are as defined in   Section \ref{iid}.
\begin{algorithm}[tb]
\caption{Backward Greedy Selection (BGS2)}
\label{alg3}
\begin{algorithmic}[1]
\STATE Input: users set $\mathcal{U}^i$, and corresponding matrix $P_{\mathcal{U}^i}$.
\STATE Output: users subset ${\mathcal{U}_B^{ALG}}$.
\STATE Initialize $P_{B}=P_{\mathcal{U}^i}$, $C_{B}=C_{\mathcal{U}^i}$, and $\mathcal{U}_B^{ALG} =\mathcal{U}^i$.
\FOR {$j=1$ to $ | \mathcal{U}^i | - B$}
\STATE 
$v_j \leftarrow \arg\min_{v \in \mathcal{U}_B^{ALG} } \Big\{  \text{tr} \Big( \left( P_{B \setminus v} C_{B \setminus v}^{-2} P_{B \setminus v} ^{ \top} \right)^{-1}   \Big)  \Big\} 
$
\STATE Update $P_{B} \leftarrow P_{B \setminus v_j} $ and $C_{B} \leftarrow C_{B \setminus v_j} $. 
\STATE Update $\mathcal{U}_B^{ALG} \leftarrow \mathcal{U}_B^{ALG} \setminus v_j $.
\ENDFOR
\end{algorithmic}
\end{algorithm}
For Algorithm \ref{alg3} we can prove the following:
\begin{thm} \label{greedy2}
Let $i$ be a new item, and $\mathcal{U}^i $ its corresponding set of available users, and assume that  $ \mathbb{E} \left[ \varepsilon_{ui} \right] =0 $ and $ \mathbb{E} \left[ \varepsilon_{ui}^2 \right] = \sigma_u^2 $ for all $u \in \mathcal{U}$. Then, Algorithm \ref{alg3} generates a subset $\mathcal{U}^{ALG}_B \subset  \mathcal{U}^i$, for which the following holds:
\begin{align*} \label{approx}
\E \left[ \text{MSE} \big(  \mathcal{U}^{ALG}_B \big) \right] &  \leq \frac{1 }{ \left| \mathcal{U} \right| }  \sum_{u \in \mathcal{U} } \sigma_u^2   \\
& +   \frac{e^t-1}{t} \left( \E \left[  \text{MSE} \big(  \mathcal{U}^{*}_B \big) \right] -  \frac{1 }{ \left| \mathcal{U} \right| }  \sum_{u \in \mathcal{U} } \sigma_u^2 \right).
\end{align*}
\end{thm}
The proof relies on the same techniques of Theorem \ref{GREEDY}, albeit plugging the variance-scaled matrix $ P_B C_B^{-1}$ instead of the matrix $P_B$, and is thus omitted here.\\

\textit{\textbf{Reduction to unknown variances.}}
Until now, we assumed that each user has its own noise variance with respect to the model. Our analysis relied on the fact that these variances are known to us in advance. Clearly, this is never the case in real life.  However, in the setting of CF-based recommender systems, we can utilize the past ratings of a given user to estimate this variance. Formally, let $\mathcal{R} (u)$ be the set of items that user $u$ has rated. Then, $\sigma_u^2$ can be estimated as follows:
$$
\tilde{\sigma}_u^2 = \frac{1}{ | \mathcal{R}(u)  | } \sum_{i \in \mathcal{R}(u) } ( \tilde{r}_{ui} - {r}_{ui} )^2 ,
$$
where $ \tilde{r}_{ui} $ is as predicted by the LFM (Equation~\eqref{eq:LFM}). 
In the experimental section, we show that despite the fact that the variances are estimated and not known, Algorithm \ref{alg3} outperforms all the baselines (including Algorithm \ref{alg2}).

\vspace{\secSpace}
\section{Experimental results}
\label{sec:results}

\subsection{Dataset and Model}
\label{sec:setup}

We consider a large movie ratings dataset released by Netflix as the basis of a well publicized competition \cite{Bennett-Lanning_KDD2007}. The original dataset contains more than $100$ million date-stamped ratings (integers between $1$ and $5$) collected between Nov. 11, 1999 and Dec. 31, 2005, from about $480,000$ anonymous Netflix customers on $17,770$ movies. 

The dataset is processed as follows: a new ratings matrix (denoted henceforth by $\mathcal{R}$) is constructed from $17,000$ arbitrarily chosen movies. The matrix $\mathcal{R}$ contains about $96$ million ratings from the same number of users as in the original dataset. The other 770 movies are used as ``new items'' in the various experiments we conduct\footnote{Note we cannot use one of the official Netflix test sets, as none of them captures the notion of new items.}. Note that since most users have not rated most of the movies, the ratings matrix $\mathcal{R}$ is very sparse (only $1\%$ of all ratings exist).

We factorize $\mathcal{R}$ using the LFM (Equation \eqref{eq:LFM}), with its dimension set to be $k=20$. We use the common Root Mean Squared Error (RMSE) metric to measure the prediction accuracy of the resulted model. For the task of minimizing this metric, we apply the Stochastic Gradient Descent (SGD) algorithm, where the learning rate is inspired by the AdaGrad algorithm of  \cite{duchi2011adaptive}. 

\subsection{Experimental Setup}
We now describe how we test algorithms for user selection in an offline manner on the Netflix dataset.
As mentioned before, a set of $770$ movies is withheld from our model. We henceforth denote this set by $\mathcal{N}$, and think of it as a set of movies that are new to our recommender system. 
For each movie $i \in \mathcal{N}$, we denote by $\mathcal{R}(i)$ the set of Netflix users who have actually rated this movie. We then adapt the dataset to our setting as follows:
\begin{description}
\item[The pool of available users:] we consider the set $\mathcal{R}(i)$ as the pool of available users to rate movie $i$, i.e., we set $\mathcal{U}_i = \mathcal{R}(i)$.
Note that indeed, $\mathcal{U}^i$ changes for every movie $i$, as the actual available ratings differ from one movie to another in the Netflix dataset. This coincides with the pool of available users changing in time.
\item[Selecting a subset of users:] the conceptual step of assigning a subset of users $\mathcal{U}^i_{B} \subset  \mathcal{U}^i$ to rate movie $i$ corresponds to the action of revealing their actual ratings in the dataset. 
\item[The set of all users:] due to the small portion of actual ratings in the dataset, we cannot measure the prediction error on all users. Thus,
after revealing the actual ratings of the subset $\mathcal{U}^i_{B} $, we use the remaining unrevealed ratings for this task, i.e., we set $\mathcal{U} = \mathcal{U}_i \setminus \mathcal{U}_B^i$.  Note that we {\em  essentially use the non-selected users for the evaluation task} in this step. 
\end{description}

Despite the fact that our analysis was carried out using the MSE metric, we measure the prediction error of the different approaches using the RMSE metric:
\begin{equation*}
\text{RMSE} \big( \mathcal{U}^i_B \big) =  \sqrt{ \frac{1}{\left| \mathcal{U} \right|} \sum_{u \in \mathcal{U}} \left(  {r}_{ui} - \tilde{r}_{ui} \right)^2 }  .
\end{equation*}
Since the root square is a monotone function, minimizing the MSE is equivalent to minimizing the RMSE (in the sense that both problems has the same minimizer). Thus, to be consistent with previous works, our experimental results are presented using the RMSE metric.

An important aspect of our experimental setup is the preselection of users who watched the ``new'' movie. This might seem unrealistic, since generally there is no guarantee that users will rate the new item (as we implicitly assume). We justify our setup as follows: each user we select to rate the new item is given some reward, which is (in most cases) sufficient to convince the user to provide rating. This mitigates the need to select ``frequent raters'', and highlights the rather interesting geometric aspect of the problem:  selected users should be  scattered (in a sense) in the LFM space.

\subsection{Baselines}
\label{baseline}


We turn to present several common and less common baselines from the recent literature. Recall that our problem (as defined in Equation~\eqref{prob_def}) was divided into two separate problems, and thus we consider two types of baselines. We begin with baselines for \textbf{Problem A}: rating prediction. \\

\textit{\textbf{Similarity based estimator.}}
Just like the least squares estimator, this baseline seeks to estimate $( b_i , Q_i ) $ and consequently generate a prediction $\tilde{r}_{ui} =  \mu + \tilde{b}_i  + b_u+  \tilde{Q}_i  P_u $.
 The estimation uses the assigned users who liked the item, and relies on 
the following underlying assumption: ``good'' user-item interaction is expressed by the closeness of the corresponding representative vectors in the LFM space. 
Thus, yielding the estimator:
 \begin{align*} 
 \tilde{b}_i & =   \frac{1}{| \mathcal{U}^i_B |}   \sum_{v \in  \mathcal{U}^i_B  } \left( r_{vi}  - b_v \right) - \mu \ , \\
\tilde{Q}_i & =   \frac{1}{ \left| \{ v \in \mathcal{U}^i_B | r_{vi} \geq \gamma \} \right| }   \sum_{v \in \mathcal{U}^i_B | r_{vi} \geq \gamma} P_v  \ ,
\end{align*}
where $\gamma$ is a predefined threshold (e.g., $4$ or $5$ for the Netflix dataset). We note that integrating negative feedback (low ratings) yields an empirically inferior similarity based estimator that is thus omitted here.
This approach was successfully practiced in \cite{Paterek_KDD2007,Koren_KDD2008,Aharon-Kagian-Koren-Lempel_RecSys2012,Aizenberg-Koren-Somekh_WWW2012}.\\
 
 \vspace{-0.14cm}

We continue with presenting baselines for \textbf{Problem B}: user selection.\\

\vspace{-0.14cm}

\textit{\textbf{Forward greedy selection (B1).}}
This baseline relies as well on the optimal design approach, and specifically on the A-optimality criterion. Basically, it aims to minimize the trace of the inverse of the information matrix $P_B P_B^\top$, similarly to our first algorithm. However, the selection is done in a forward manner: starting with an empty set, and greedily adding the user whose contribution is the highest. This approach was successfully practiced in many fields of active learning, rather than only for the item cold-start problem of CF-based recommenders (see \cite{krause2007near} for a more comprehensive survey). However, we are not aware of any performance guarantee for this algorithm. \\

\vspace{-0.14cm}

\textit{\textbf{Clustering (B2).}}
 Intuitively, having a heterogeneous set of users will result in a more objective ratings of the item. To this end, we suggest to cluster the users with respect to their representation in the Euclidean LFM space and subsequently sample from the resulting clusters. For the clustering task, we use the standard $k$-means algorithm with cosine similarity as its similarity measure.

We propose two clustering-based baselines for the selection problem, each differs in the number of clusters and the sampling procedure: (1) split  the pool of available users into $B$ clusters, and select the users 
that are closest to the cluster center of mass (one user per cluster); and (2) split the pool of available users into $c$ clusters (where $c<B$), and from each cluster select number of users at random and proportionally to the cluster size. For this baseline, the value of $c$ is determined empirically.\\

\vspace{-0.14cm}

\textit{\textbf{Random selection (B3).}}
In this baseline, the $B$ users are selected randomly  from the pool of available users. 
At first glance, this baseline seems rather weak. However, if $B$ is sufficiently large, then $\mathcal{U}^i_B $ is a statistically good representative of $ \mathcal{U}^i$, and this baseline cannot be easily beaten. Moreover, this baseline has many practical appeals and is thus very popular with state-of-the-art recommenders.\\

\vspace{-0.14cm}

\textit{\textbf{Frequent raters (B4).}}
Here we select users who tend to provide many ratings. This baseline utilizes the following conjecture: users with extensive rating history have better modeling in CF-based recommender systems. Therefore, preferring these users over users that have limited rating history seems reasonable in our setting. \\

\vspace{-0.14cm}

\textit{\textbf{Edgy raters (B5).}}
This baseline selects users that provide diverse ratings, or more accurately, ratings with large variance. Intuitively, these users provide more information regarding an item they like or dislike, than users that tend to rate all items in a similar manner. \\

\vspace{-0.14cm}

\textit{\textbf{Early birds raters (B6).}}
This baseline slightly differs from the previous baselines: instead of actively selecting users to rate the new item, we 
conceptually invite all users to opine on the new item at their convenience, and  
consider the (chronologically) first $B$ returned ratings as our selection. The intention here is not to offer another active selection baseline, but rather to study the utility of early reviews, coming from people who are enthusiastic about trying out the new item, for modeling the rest of the population.

\subsection{Results}
We separately report the results for each of the considered problems, starting with the rating prediction problem and moving on to the user selection problem.

\subsubsection{Rating Prediction Comparison}
We compare the performance of the least squares estimator and the similarity based estimator.
To this end, we select 300 movies from the set $\mathcal{N}$, arbitrarily.  For each of these movies, we randomly select $B$ users ($B=2,4,\ldots,100$), and apply both estimators to generate predictions for the remaining users. The prediction accuracy is measured using the RMSE metric over all predictions (and not separately for each movie).
This routine is averaged over 50 runs to ensure stability. 
\begin{figure}[tb] 
  \centering
    \includegraphics[width=0.45 \textwidth]{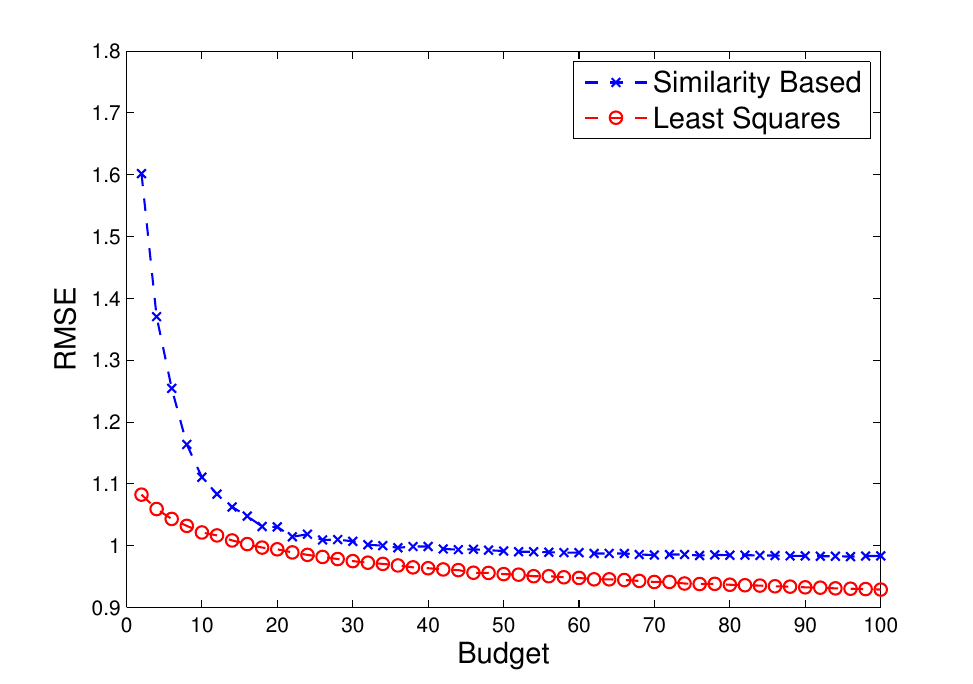}
     \caption{\textbf{Rating prediction comparison.} }
          \label{fig1}
\end{figure}

As evident in Figure \ref{fig1}, the least squares estimator outperforms the considered baseline. 
Qualitatively, similar results were obtained when applying the other (non-random) selection schemes of $B$ users, and are thus omitted here. 

\subsubsection{User Selection Comparison}
Here the performance of our approach is compared with the baselines (presented in Section \ref{baseline}) for \textbf{Problem B}. To carry out  the experiment, we select 300 movies from the set $\mathcal{N}$, arbitrarily.  For each of these movies, we apply all six baselines (B1-B6) and our two algorithms (denoted by BGS1 and BGS2) to select a subset of $B$ users. We then generate $( \tilde{b}_i , \tilde{Q}_i )$ using the least squares estimator (or the generalized least squares estimator for BGS2), regardless of how the selection was made. Since this estimator outperforms the similarity based estimator for any selection of $B$ users (as demonstrated in Figure \ref{fig1}), we do not lose information by omitting them. Here again, we measure the prediction accuracy using the RMSE metric over all predictions. For the non-deterministic baselines: Clustering (B2) and Random (B3), we average the results over 50 runs  to ensure stability.

\begin{figure*}[tb]
  \centering
    \includegraphics[width= \textwidth]{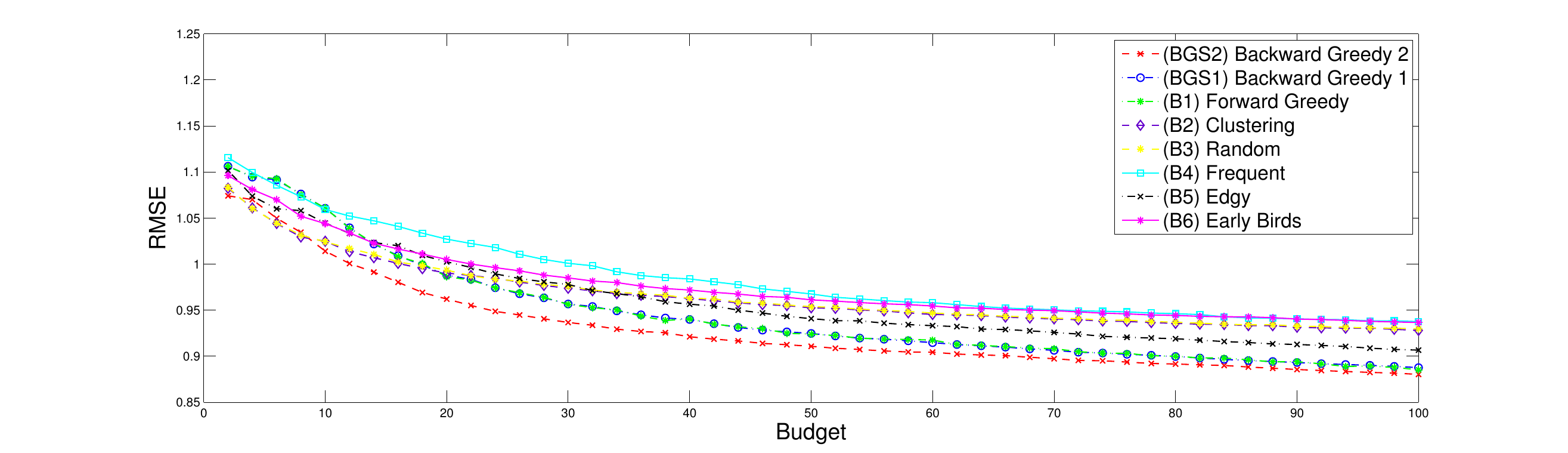}
     \caption{User selection comparison.} 
          \label{fig2}
\end{figure*}

As can be seen in Figure \ref{fig2}, the optimal design based algorithms outperform the other baselines for almost any budget. 
In particular, BGS2 (which  accounts for independent yet not identically distributed noises) significantly
\footnote{The results are significant at any significance level, since we average them over $\left| \mathcal{N} \right| = 300$ movies.} 
surpasses BGS1 and B1 (which are practically indistinguishable) for any budget. The second best performing baseline is B5, which utilizes the ratings of the ``edgy'' users. We conjecture that such users provide more information regarding a new item, in addition to being well-scattered in the LFM space. Baselines B2 and B3 (Random and Clustering\footnote{We present here only the better performing clustering technique, which is the second one discussed in Section \ref{baseline}.}) seem to perform similarly in our setting, yet a closer look discovers that by clustering the users in the LFM space we get more stable results (in the sense of lower variance of the resulted RMSE). We remind the reader that our basic CF model was trained using state-of-the-art algorithms, where every percent of improvement is not easily achieved.  
%


\vspace{\secSpace}
\section{Conclusions and Future Work}
\label{sec:conc}

In this work we started with a mature LFM model and tried to tackle the item cold-start problem by answering two questions: (a) given ratings of $B$ users, how can we estimate other users' ratings of an arbitrary new item without retraining the model? and (b) how to choose the set of $B$ users? We showed that in order to get approximately optimal results in terms of mean squared error, the above questions should be tackled jointly. In particular, we applied optimal design techniques and devised two greedy algorithms that achieve that goal under certain assumptions. We used the Netflix dataset to demonstrate the superiority of the proposed algorithms over a set of previously considered and non-considered baselines.

Our work can be seen as a ``one-shot'', non-adaptive active learning scheme that selects the $B$ users at once. This leaves three interesting variants for future work.
The first variant we consider is an {\em adaptive active learning} scheme, in which users are chosen sequentially, with each selection done only after receiving the ratings of the previously chosen user.
The second variant tackles an {\em online} version of the problem, where potential users arrive one by one and the publisher must decide whether to assign them the item or not. In several recommendation settings, e.g. news recommendations, new items constantly arrive and have very short lifetimes, thereby requiring that the publisher identifies quickly which users to recommend them to. 
The tradeoff between (a) selecting a good set of users, and (b) quickly completing the reviewing process of the item, presents a challenging optimization problem. The third variant discusses a {\em multi-item} setting, where the publisher must concurrently explore several new items while users can review only one or a small fixed number of items. Note that the multi-item setting can be investigated in both the offline and the online settings.

In addition to the above variants, an interesting direction would be to study the user cold-start problem using the techniques we presented in this work. As discussed in the related work section, the user and the item cold-start problem are essentially different, mainly due to the ability to adaptively interview new users in order to bootstrap their modeling. Whereas this enables decision tree type solutions for user cold-start handling (which are currently the state-of-the-art), such solutions cannot be implemented in the setting considered in this paper. However, applying the optimal design approach to the user cold-start problem is possible and may potentially lead to better theoretical and empirical guarantees. 


\bibliographystyle{plain}
\bibliography{itemColdStart}  
%
%
\newpage
\appendix
\section{Proofs} \label{app}

\subsection{Proof of Lemma \ref{TRACE}}
\label{app: proof of lemma 1}
%
Let $i$ be a new item, and $\mathcal{U}^i $ its corresponding set of available users. 
Now, given a set of $B$ users $ \mathcal{U}_{B}^i \subset  \mathcal{U}^i $, we can estimate $ ( b_i , Q_i ) $ using the least squares estimator:
\begin{align}   \label{analytic_solution}  
 \big( \tilde{b}_i , \tilde{Q}_i \big) 
& = \left(  \sum_{ v \in \mathcal{U}^i_B } {P}^\prime_v {P}_v^{\prime \top} \right)^{-1} \left( \sum_{ v \in \mathcal{U}^i_B } (r_{vi} - b_v - \mu )   {P}_v^{\prime} \right)   \nonumber\\
& =  \left( P_B P_B^{ \top} \right)^{-1}  P_B r_B .
\end{align}
Notice that according to Equation \eqref{eq:LFM_ground} we can also express the vector $r_B$ as follows:
$$ r_B  =   {P}_B^\top  \big( b_i , Q_i \big)   + \varepsilon_{B} ,$$
and by substituting the above in Equation \eqref{analytic_solution} and shifting sides we get that
\begin{equation} \label{dist}
  \big(  \tilde{b}_i , \tilde{Q}_i \big) -  \big( b_i , Q_i  \big)  =    \left( P_B P_B^{ \top} \right)^{-1} P_B   \varepsilon_B .
\end{equation}
We denote by $\text{MSE} (  \mathcal{U}^i_B ) $  the resulted MSE value by using the ratings of the users within $  \mathcal{U}^i_B $ in our prediction model, that is, $ \text{MSE} \big(  \mathcal{U}^i_B \big) = \frac{1}{\left| \mathcal{U} \right|} \sum_{u \in  \mathcal{U}  } \left( \tilde{r}_{ui} - {r}_{ui} \right)^2 $.

Then, we can derive:
\begin{align*} 
\E \Big[ \text{MSE} \big(  \mathcal{U}^i_B \big) \Big]  &   = \frac{1}{\left| \mathcal{U} \right|} \E \left[ \sum_{u \in  \mathcal{U}  } \left( \tilde{r}_{ui} - {r}_{ui} \right)^2 \right]   \nonumber\\
& = \frac{1}{\left| \mathcal{U} \right|}  \E \left[  \left\| P^\top \left( ( \tilde{b}_i , \tilde{Q}_i)  - ( {b}_i , {Q}_i) \right)  + \varepsilon_{\mathcal{U}} \right\| _2^2 \right] \nonumber\\
& \stackrel{(1)}{=} \frac{1}{\left| \mathcal{U} \right|} \E \Bigg[  \left\|  P^\top  \left( P_B P_B^{ \top} \right)^{-1} P_B   \varepsilon_B  +  \varepsilon_{\mathcal{U}} \right\| _2^2 \Bigg] \nonumber\\
&  \stackrel{(2)}{=} \frac{1}{\left| \mathcal{U} \right|} \E \Bigg[  \left\| P^\top  \left( P_B P_B^{ \top} \right)^{-1} P_B   \varepsilon_B   \right\| _2^2 + \left\| \varepsilon_{\mathcal{U}} \right\| _2^2  \Bigg] ,
\end{align*}
where the expectation is taken over the noise $\eps_{vi}$ and $\eps_{ui}$. Equality (1) follows by substituting Equation \eqref{dist}; equality (2) holds since $ \mathbb{E} \left[ \varepsilon_{ui} \right] =0 $ for all $ u \in \mathcal{U} $, and also since $\varepsilon_B$ and $\varepsilon_\mathcal{U}$ are independent.

Next, we  use the assumption  $ \mathbb{E} \left[ \varepsilon_{ui}^2 \right] = \sigma^2 $  to get:
\begin{align*} 
\mathbb{E} \left[ \text{MSE} \big(  \mathcal{U}^i_B \big) \right] &   =  \frac{1}{\left| \mathcal{U} \right|} \mathbb{E} \left[  \left\| P^\top  \left( P_B P_B^{ \top} \right)^{-1} P_B   \varepsilon_B  \right\| _2^2 \right]  +  \sigma^2 \nonumber\\
&  \stackrel{(3)}{=} \frac{1}{\left| \mathcal{U} \right|} \left\|    P^\top  \left( P_B P_B^{ \top} \right)^{-1} P_B C_B  \right\|_F^2 +\sigma^2 \nonumber\\
&  \stackrel{(4)}{=}  \sigma^2   \text{tr} \left( \left( P_B P_B^{\top} \right)^{-1}  \right) +\sigma^2  ,
\end{align*}
where $ C_B $ denotes the square root of the covariance matrix of $ \varepsilon_{B} $; equality (3) follows from the definition of the Frobenius norm $ \|   \|_F $, and equality (4) holds since $C_B = \sigma^2   I_{B \times B}$, and since we assume that $ P P ^\top = \left| \mathcal{U} \right|   I_{(k+1) \times (k+1)} $. 

Thus, the lemma is obtained.

\subsection{Proof of Lemma \ref{SECOND}}
\label{app: proof of lemma 4}

%
By taking expectation we have that:
\begin{align*}
\mathbb{E} \left[  \big( \tilde{b}_i , \tilde{Q}_i \big)   \right] &= \mathbb{E} \left[ \left( P_B C_B^{-2} P_B^{ \top} \right)^{-1}  P_B C_B^{-2}  r_B  \right] \\
& = \mathbb{E} \left[  \left( P_B C_B^{-2} P_B^{ \top} \right)^{-1}  P_B C_B^{-2} {P}_B^\top  \big( b_i , Q_i \big)   \right] \\
& \quad + \mathbb{E} \left[ \left( P_B C_B^{-2} P_B^{ \top} \right)^{-1}  P_B C_B^{-2}  \varepsilon_{B} \right] \\
& =   \big( {b}_i , {Q}_i \big)  +  \left( P_B C_B^{-2} P_B^{ \top} \right)^{-1}  P_B C_B^{-2}  \mathbb{E} \big[ \varepsilon_B \big]   \\
& =  \big( {b}_i , {Q}_i \big) ,
\end{align*}
where in the second equality  we used the assumption that $r_B  =   {P}_B^\top  \big( b_i , Q_i \big)   + \varepsilon_{B}$.

\subsection{Proof of Lemma \ref{trace2}}
\label{app: proof of lemma 5}

%
%
%
Recall that as before we assume that
$$ P P ^\top = \left| \mathcal{U} \right|   I_{(k+1) \times (k+1)} \ .$$
Now, let $i$ be a new item, and $ \mathcal{U}_{B}^i \subset  \mathcal{U}^i $ a subset of users assigned for the task of rating it.
Then, we can  substitute $r_B$ in the generalized least squares estimator (Equation \eqref{generalized}) to get the following result:
$$
\big( \tilde{b}_i , \tilde{Q}_i   \big)  - \big( {b}_i , {Q}_i   \big)  =  \left( P_B C_B^{-2} P_B^{ \top} \right)^{-1}  P_B C_B^{-2}  \varepsilon_{B}  .
$$
We substitute the above in the MSE equation, and rewrite our motivation problem:
\begin{align*} 
& \E \left[ \text{MSE} \big(  \mathcal{U}^i_B \big) \right] \\
&   = \frac{1}{\left| \mathcal{U} \right|} \E \left[  \sum_{u \in  \mathcal{U}  } \left( \tilde{r}_{ui} - {r}_{ui} \right)^2 \right]  \nonumber\\
& = \frac{1}{\left| \mathcal{U} \right|}  \E \left[  \left\|  P^\top \left( ( \tilde{b}_i , \tilde{Q}_i)  - ( {b}_i , {Q}_i) \right)  + \varepsilon_{\mathcal{U}} \right\| _2^2  \right] \nonumber\\
& = \frac{1}{\left| \mathcal{U} \right|} \E \left[  \left\| P^\top  \left( P_B C_B^{-2} P_B^{ \top} \right)^{-1}  P_B C_B^{-2}  \varepsilon_{B}  + \varepsilon_{\mathcal{U}} \right\| _2^2\right]   \nonumber\\
& = \frac{1}{\left| \mathcal{U} \right|} \E \left[   \left\|   P^\top  \left( P_B C_B^{-2} P_B^{ \top} \right)^{-1}  P_B C_B^{-2}  \varepsilon_{B}   \right\| _2^2 \right] +  \frac{  \E \left[  \left\| \varepsilon_{\mathcal{U}} \right\| _2^2  \right]  }{ \left| \mathcal{U} \right| } \\
&= \frac{ 1 }{\left| \mathcal{U} \right|}  \mathbb{E} \left[ \left\|   P^\top  \left( P_B C_B^{-2} P_B^{ \top} \right)^{-1}  P_B C_B^{-2}  \varepsilon_{B}   \right\| _2^2 \right] + \frac{1 }{ \left| \mathcal{U} \right| }  \sum_{u \in \mathcal{U} } \sigma_u^2 
 \nonumber\\
& = \frac{1}{\left| \mathcal{U} \right|}  \left\| P^\top  \left( P_B C_B^{-2} P_B^{ \top} \right)^{-1}  P_B C_B^{-1}     \right\|_F^2 +  \frac{1 }{ \left| \mathcal{U} \right| }  \sum_{u \in \mathcal{U} } \sigma_u^2  \nonumber\\
&  =  \text{tr} \left( \left( P_B C_B^{-2} P_B^{\top} \right)^{-1}  \right) +  \frac{1 }{ \left| \mathcal{U} \right| }  \sum_{u \in \mathcal{U} } \sigma_u^2   ,
\end{align*}
where the equalities follow similarly to Lemma \ref{TRACE}.

\end{document}